\documentclass[12pt,a4paper]{article}
\usepackage{times}
\usepackage{fullpage}
\usepackage{graphicx}
\usepackage{amsthm}
\usepackage{amsmath}
\usepackage{amsfonts}
\usepackage{amssymb}
\usepackage{dsfont}
\usepackage{setspace}
\usepackage{algorithm}
\usepackage{algorithmic}
\usepackage{url}
\usepackage{subfigure}
\usepackage{paralist}
\usepackage{multirow}
\IfFileExists{srcltx.sty} {\usepackage[active]{srcltx}}

\DeclareMathOperator*{\argmin}{arg\,min}

\newcommand{\abs}[1]{\left\vert#1\right\vert}
\newcommand{\set}[1]{\left\{#1\right\}}

\newcommand{\girth}{\mathrm{girth}}%
\newcommand{\eqdf }{\triangleq}%
\newcommand{\R }{\mathds{R}}
\newcommand{\E }{\mathds{E}}
\newcommand{\N }{\mathds{N}}
\newcommand{\calN }{\mathcal{N}} 
\newcommand{\calC }{\mathcal{C}} 
\newcommand{\calV }{\mathcal{V}} 
\newcommand{\calJ }{\mathcal{J}} 

\newcommand{\calB }{\mathcal{B}}
\newcommand{\calP }{\mathcal{P}}

\renewcommand{\Pr }{\mathrm{Pr}}

\newcommand{\cost }{\mathrm{cost}}

\newcommand{\LP}{\textsc{lp}}
\newcommand{\ML}{\textsc{ml}}

\newcommand{\NWMS}{\textsc{nwms}}

\newcommand{\calCj }{\mathcal{C}^j}
\newcommand{\calCbar }{\overline{\mathcal{C}}}
\newcommand{\calCbarJ}{\overline{\mathcal{C}}^{\mathcal{J}}}
\newcommand{\calCJ }{\mathcal{C}^\mathcal{J}}

\newenvironment{remark}[0]{\noindent \textbf{   Remark:}~}{\endproof}

\newtheorem{definition}{Definition}
\newtheorem{lemma}[definition]{Lemma}

\newtheorem{proposition}[definition]{Proposition}
\newtheorem{theorem}[definition]{Theorem}
\newtheorem{corollary}[definition]{Corollary}

\begin{document}

\title{Local-Optimality Guarantees for Optimal Decoding Based on Paths}

\author{
      Guy Even \thanks{School of Electrical Engineering, Tel-Aviv University, Tel-Aviv 69978, Israel.  \mbox{{E-mail}:\ {\tt guy@eng.tau.ac.il}.}}
      \and
      Nissim Halabi \thanks{School of Electrical Engineering, Tel-Aviv University, Tel-Aviv 69978, Israel. \mbox{{E-mail}:\ {\tt nissimh@eng.tau.ac.il}.}}}

\date{}

 \maketitle

\begin{abstract}
This paper presents a unified analysis framework that captures
  recent advances in the study of local-optimality characterizations
  for codes on graphs. These local-optimality characterizations are
  based on combinatorial structures embedded in the Tanner graph of
  the code. Local-optimality implies both unique maximum-likelihood (ML)
  optimality and unique linear-programming (LP) decoding optimality. Also, an
  iterative message-passing decoding algorithm is guaranteed to find
  the unique locally-optimal codeword, if one exists.

  We demonstrate this proof technique by considering a definition of
  local-optimality that is based on the simplest combinatorial
  structures in Tanner graphs, namely, paths of length $h$.  We apply
  the technique of local-optimality to a family of Tanner codes.  Inverse
  polynomial bounds in the code length are proved on the word error
  probability of LP-decoding for this family of Tanner codes.
\end{abstract}

\section{Introduction} \label{sec:intro}

The method of decoding error correcting codes based on a linear programming (LP) relaxation of maximum-likelihood (ML) decoding was introduced by Feldman, Wainwright and Karger~\cite{FK02,FWK05}. Successful decoding by LP-decoding was analyzed for different codes and channels (see e.g., \cite{FMSSW07,FS05,DDKW08,Ska11}). Significant  advances in the analysis of successful LP-decoding have been achieved recently.  Following Koetter and Vontobel~\cite{KV06}, Arora \emph{et al.}~\cite{ADS09} obtained improved bounds
for $(3,6)$-regular low-density parity-check (LDPC) codes over the binary symmetric channel (BSC). These techniques were
extended to memoryless binary-input output-symmetric (MBIOS) channels~\cite{HE11} and to Tanner codes~\cite{EH11}. The proofs in these papers are based on complicated
graphical structures and on a sophisticated analysis of a random min-sum
process. Our goal in this paper is to present a simple analysis based
on this proof technique that proves nontrivial bounds for a broad
family of Tanner codes.

\noindent
The proof technique in~\cite{ADS09} is based on the following steps:
\begin{enumerate}
\item Define a set of deviations. A deviation is induced by
  combinatorial structures in the Tanner graph or the computation
  tree~\cite{Wib96, KV06,ADS09,Von10,EH11}.
\item Define local-optimality. Local-optimality is a combinatorial characterization of a codeword $x\in\{0,1\}^N$ with respect to a channel output $\lambda\in\R^N$. This definition is based on: (1)~a
  definition of a relative point for a codeword $x$ and each
  deviation $\beta$~\cite{Fel03}, and (2)~a definition of the cost of each relative point with
  respect to the log-likelihood ratios (LLR) vector $\lambda$.  Loosely speaking, a codeword
  $x$ is locally-optimal if its cost is smaller than the cost of every
  relative point.
\item Prove that if $x$ is locally-optimal codeword w.r.t. an LLR vector $\lambda$, then $x$ is the unique
  maximum-likelihood (ML) codeword.  The proof is based on a
  decomposition lemma that states that every codeword is a conical sum
  of deviations.
\item Prove that if $x$ is locally-optimal codeword w.r.t. an LLR vector $\lambda$, then $x$ is the unique
  linear-programming (LP) codeword. This proof is based on a lifting
  lemma~\cite{HE11} that states that local-optimality is invariant under liftings
  of codewords to covering graphs.
\item Analyze the probability that there does not exist a locally-optimal codeword.
\end{enumerate}

We demonstrate this proof technique by considering the simplest
combinatorial structures in Tanner graphs, namely, paths of length
$h$.  We attach to each path $p$ a deviation $\beta\in \R^N$ that
equals the multiplicity of each variable node along $p$ divided by its
degree times $h+1$. Surprisingly, all the ingredients of the proof
technique can be demonstrated with respect to this primitive set of
deviations.

The family of codes we consider is the set of Tanner codes that
satisfy two properties: (1)~the local codes contain only codewords of
even weight, and (2)~all the variable node degrees are even. We refer
to this family of Tanner codes as \emph{even Tanner
  codes}\footnote{The parity of the node degrees is used to prove that
  the subgraph induced by codewords is Eulerian, a key component in the decomposition lemma of step $3$.}.  Among the codes in
this family are:
\begin{inparaenum}[(1)]
\item LDPC codes with even left degrees,
\item irregular repeat accumulate codes where the repetition factors are even, and
\item expander codes with even variable node degrees and even weighted local codes.
\end{inparaenum}

Apart from offering a simple application of a powerful proof
technique, we shed light on two issues.  (1)~The deviations should
\emph{not} be limited by the girth of the Tanner graph.  Indeed, we
consider paths of arbitrary length $h$ which need not be simple for
steps 1-4 of the proof technique.  One should note that step 5 that
bounds the probability of the existence of a locally-optimal codeword
requires independence of the LLRs in the deviation. Hence, the bound
here holds only for simple paths, thus limiting $h$ by the girth.
(2)~The Tanner graph need not be regular. Irregular degrees are handled
by dividing the deviations by the node degrees (see~\cite{Von10}).

Local-optimality is also related to iterative decoding.  An iterative
message-passing decoding algorithm is presented in~\cite{EH11} with
the guarantee that, if a locally-optimal codeword exists, then the
decoding algorithm finds it.  Moreover, since local-optimality can be
verified efficiently, an ML-certificate is obtained whenever there
exists a locally-optimal codeword.

We prove inverse polynomial bounds in the code length on the word
error probability for LP-decoding of even Tanner codes whose Tanner
graphs have logarithmic girth.  For certain sub-classes of even Tanner
codes this technique provides new bounds on the word error probability
of LP-decoding.  In the case of repeat accumulate codes, we provide a
simple proof of previous results~\cite{FK02,GB11}.

\paragraph{Organization.} The remainder of the paper is organized as follows. Section~\ref{sec:prelim} provides background on ML-decoding and LP-decoding of Tanner codes over MBIOS channels. Section~\ref{sec:LOCert} contains steps 1-4 in the proof technique of local-optimality applied to even Tanner codes. In Section~\ref{sec:bounds} we demonstrate step 5 by a simple probabilistic analysis of the event where no locally-optimal codeword exists. We conclude with a discussion in Section~\ref{sec:discussion}.

\section{Preliminaries} \label{sec:prelim}

\paragraph{Graph Terminology.}
Let $G=(V,E)$ denote an undirected graph.  Let $\calN_G(v)$ denote the set
of neighbors of node $v \in V$. Let $\deg_G(v)\triangleq\lvert\mathcal{N}_G(v)\rvert$ denote the edge degree of node $v$ in graph $G$.
A path
$p=(v,\ldots,u)$ in $G$ is a sequence of vertices such that there
exists an edge between every two consecutive nodes in the sequence
$p$. A \emph{simple} path is a path with no repeated vertex. A \emph{simple cycle} is a closed path where the only repeated vertex is the first and last vertex. A path $p$ is \emph{backtrackless} if every three consecutive vertices along $p$ are distinct (i.e., a subpath $(u,v,u)$ is not allowed).
Let $\lvert p\rvert$ denote the length of a path $p$, i.e., the
number of edges in $p$.
Let $\girth(G)$ denote the length of the shortest cycle in $G$.
The \emph{subgraph of $G$ induced by $S \subseteq V$} consists of $S$ and all edges in $E$, both endpoints of which
are contained in $S$. Let $G_S$ denote the subgraph of $G$ induced by $S$.

\paragraph{Tanner-codes and Tanner graph representation.}
Let $G=(\mathcal{V} \cup \mathcal{J}, E)$ denote an edge-labeled
bipartite-graph, where $\calV = \{v_1,\ldots,v_N\}$ is a set of
$N$ vertices called \emph{variable nodes}, and $\mathcal{J} =
\{C_1,\ldots,C_J\}$ is a set of $J$ vertices called \emph{local-code
  nodes}.
We associate with each local-code node $C_j$ a linear code $\calCbar^j$ of length $\deg_G(C_j)$. Let $\calCbarJ \triangleq \big\{\calCbar^j\ :\ 1\leqslant j \leqslant J \big\}$ denote the set of \emph{local codes}, one for each local-code node. We say that $v_i$ \emph{participates} in $\calCbar^j$ if $(v_i,C_j)$ is an edge in $E$.

A word $x = (x_1,\ldots,x_N) \in \{0,1\}^N$ is an assignment
to variable nodes in $\calV$ where $x_i$ is assigned to $v_i$. Let $\mathcal{V}_j$ denote the set $\calN_G(C_j)$ ordered according to labels of
edges incident to $C_j$. Denote by $x_{\mathcal{V}_j} \in
\{0,1\}^{\deg_G(C_j)}$ the projection of the word $x = (x_1,\ldots,x_N)$ onto
entries associated with $\mathcal{V}_j$.

The \emph{Tanner code} $\calC(G,\calCbarJ)$ based on the labeled
\emph{Tanner graph} $G$ is the set of vectors $x \in \{0,1\}^N$
such that $x_{\mathcal{V}_j}$ is a codeword in $\calCbar^j$ for every
$j \in \{1,\ldots,J\}$.
Let $\calCj$ denote the \emph{extension} of the local code $\calCbar^j$ from length $\deg(C_j)$ to length $N$ defined by $\calCj \triangleq \{x \in \{0,1\}^N\mid x_{\calV_j} \in \calCbar^j \}$. The Tanner code is simply the intersection of the
extensions of the local codes, i.e., $\calC(G,\calCbarJ) = \bigcap_{j \in \{1,\ldots,J\}}{\calCj}$.

\medskip\noindent
We consider a family of Tanner codes defined as follows.
\begin{definition}[even Tanner codes] \label{def:evenTannerCodes}
A Tanner code $\calC(G,\calCbarJ)$ based on a Tanner graph $G=(\calV\cup\calJ,E)$ is called an \emph{even Tanner code} if:
\begin{inparaenum}[(1)]
\item  $\deg_G(v)$ is even for every $v\in\calV$, and
\item every codeword in each local code $\calCbar^j\in\calCbarJ$ has  even weight.
\end{inparaenum}
\end{definition}

\paragraph{LP decoding of Tanner codes over memoryless channels.}
Let $c_i \in \{0,1\}$ denote the $i$th transmitted binary symbol (channel input), and let $y_i\in\R$ denote the $i$th received symbol (channel output).
A \emph{memoryless binary-input output-symmetric} (MBIOS) channel is defined by a conditional probability density function $f(y_i|c_i=a)$ for $a\in\{0,1\}$, that satisfies $f(y_i|0) = f(-y_i|1)$. The binary erasure channel (BEC), binary symmetric channel (BSC) and binary-input additive white Gaussian noise (BI-AWGN) channel are examples for MBIOS channels. In MBIOS channels, the \emph{log-likelihood ratio} (LLR) vector $\lambda \in \R^N$ is defined by $\lambda_i (y_i) \triangleq
\ln\big(\frac{f(y_i|c_i=0)}{f(y_i|c_i=1)}\big)$ for every
input bit $i$. For a code $\calC$, \emph{Maximum-Likelihood (ML)
decoding} is equivalent to
\begin{equation} \label{eqn:MLdecoding}
 \hat{x}^{\ML}(y) = \argmin_{x \in \mathrm{conv}(\mathcal{C})} \langle
\lambda(y) , x \rangle,
\end{equation}
where $\mathrm{conv}(\mathcal{C})\subset[0,1]^N$ denotes the convex hull of the codewords in $\mathcal{C}$.

In general, solving the optimization problem in (\ref{eqn:MLdecoding})
for linear codes is intractable~\cite{BMT78}. Feldman \emph{et
  al.}~\cite{Fel03,FWK05} introduced a linear programming relaxation
for the problem of ML decoding of Tanner codes with single parity-check codes acting as local codes. We consider an extension of this definition to the case in which
the local codes are arbitrary as follows. The \emph{generalized
  fundamental polytope} $\calP \triangleq \calP(G,\calCbarJ)$ of a
Tanner code $\calC = \calC(G,\calCbarJ)$ is defined by
\begin{equation}
\calP \triangleq \bigcap_{\calCj \in \calCJ}{\mathrm{conv}(\calCj)}.
\end{equation}

Given an LLR vector $\lambda$ for a received word $y$, LP-decoding is
defined by the following linear program:
\begin{equation} \label{eqn:LPdecoding} \hat{x}^{\LP}(y) \triangleq
  \argmin_{x \in \mathcal{P}(G,\calCbarJ)} \langle \lambda(y) , x
  \rangle.
\end{equation}
The difference between ML-decoding and LP-decoding is that the
fundamental polytope $\calP(G,\calCbarJ)$ may strictly contain the
convex hull of $\calC$. Vertices of $\calP(G,\calCbarJ)$ that are not
codewords of $\calC$ must have fractional components and are called
\emph{pseudocodewords}.

\section{A Local Combinatorial Certificate for an Optimal Codeword} \label{sec:LOCert}

In this section we define a simple type of local-optimality characterization that is based on backtrackless paths of arbitrary length $h$ in the Tanner graph. We prove that for codewords of even Tanner codes, this characterization suffices both for ML-optimality and LP-optimality.

\begin{definition}[normalized characteristic vector]
Consider a Tanner graph $G=(\calV\cup\calJ,E)$. The \emph{normalized characteristic vector} $\chi_G(p)\in\R^{\abs{\calV}}$ of a path $p$ is defined as follows. For every $v\in\calV$ the component $[\chi_G(p)]_v$ equals to the multiplicity of the variable node $v$ in the path $p$ divided by its degree in $G$. Formally,
\begin{equation}
[\chi_G(p)]_v = \frac{1}{\deg_G(v)}\cdot\big\lvert\set{v\mid v\in p}\big\rvert.
\end{equation}
If $p$ is closed (i.e., a cycle), then we count the multiplicity of the endpoints only once.
\end{definition}
\noindent The normalization by $\deg_G(v)$ is needed if $G$ has irregular variable node degrees.

For any fixed $h$, let $\calB^{(h)}$ denote the set of normalized characteristic vectors of backtrackless paths of length $h$ in $G$ scaled by a factor $\frac{1}{h+1}$. That is,
\begin{equation}\label{eqn:deviations}
\calB^{(h)}\triangleq\left\{\frac{\chi_G(p)}{h+1}\ \bigg\vert \ p\ \mathrm{is\ a\ backtrackless\ path\ of\ length\ } h\right\}.
\end{equation}
Vectors in $\calB^{(h)}$ are called \emph{deviations}.
Note that $\calB^{(h)}\in[0,1]^{\abs{\calV}}$ because every variable node appears less than $h+1$ times in a path of length $h$.

For two vectors $x \in \{0,1\}^N$ and $f \in [0,1]^N$, let $x\oplus f \in [0,1]^N$ denote the \emph{relative point} defined by $(x\oplus f)_i \triangleq |x_i-f_i|$~\cite{Fel03}. The following definition characterizes local-optimality based on backtrackless paths for even Tanner codes over MBIOS channels.
\begin{definition}[path-based local-optimality] \label{def:localOptimality}
Let $\mathcal{C}(G) \subset \{0,1\}^N$ denote an even Tanner code and let $h\in\N_+$. A codeword $x \in \calC(G)$ is \emph{$h$-locally optimal with respect to $\lambda \in \R^N$} if for all vectors $\beta \in \calB^{(h)}$,
\begin{equation}
\langle \lambda,x \oplus \beta \rangle > \langle \lambda, x \rangle.
\end{equation}
\end{definition}

For two vectors $y,z\in \R^N$, let ``$\ast$'' denote coordinate-wise multiplication, i.e., $(y\ast z)_i \triangleq y_i\cdot z_i$. For a word $x\in\{0,1\}^N$, let $(-1)^x\in\{\pm1\}^N$ denote the vector whose $i$th component equals $(-1)^{x_i}$.
The following proposition and corollary state that the mapping
$(x,\lambda)\mapsto(0^N,(-1)^x\ast\lambda)$ preserves local-optimality.
\begin{proposition}{\cite{EH11}}\label{proposition:isoLO}
For every $\lambda\in\R^N$ and every $\beta\in[0,1]^N$,
\begin{equation}
\langle (-1)^x\ast\lambda,\beta\rangle  = \langle  \lambda,x\oplus\beta\rangle-\langle\lambda,x\rangle.
\end{equation}
\end{proposition}

\begin{corollary}[symmetry of local-optimality]\label{corr:LOsymmetry}
For every $x\in\calC$, $x$ is $h$-locally optimal w.r.t. $\lambda$ if and only if $0^N$ is $h$-locally optimal w.r.t. $(-1)^x\ast\lambda$.
\end{corollary}

\begin{proof}
By~Proposition~\ref{proposition:isoLO},
$\langle \lambda,x \oplus \beta \rangle - \langle \lambda, x \rangle  = \langle (-1)^x\ast\lambda,\beta\rangle$.
\end{proof}

Corollary~\ref{corr:LOsymmetry} suggests that a codeword $x$ can be
verified to be $h$-locally optimal w.r.t. a given LLR $\lambda$ by
verifying that each backtrackless path of length $h$ has positive
normalized cost w.r.t. $(-1)^x\ast\lambda$.  That is, if the minimum
normalized cost of every path with length $h$ w.r.t. $(-1)^x\ast\lambda$
is positive, then $x$ is $h$-locally optimal w.r.t. $\lambda$. A
min-cost path of length $h$ in a graph can be computed by a simple
dynamic programming algorithm (Floyd's algorithm) in time
$O(h\cdot\abs{E})$. Hence, a codeword can be efficiently verified to
be locally-optimal w.r.t. $\lambda$.

\subsection{Local-Optimality Implies ML-Optimality} \label{subsec:LOimpliesML}
In the following section we show that local-optimality is sufficient for ML-optimality (Theorem~\ref{thm:MLsufficient}). The proof of Theorem~\ref{thm:MLsufficient} is based on the representation of every codeword as a conical combination of deviations in $\calB^{(h)}$ (Corollary~\ref{cor:path_decomposition}). We first prove that every codeword in an even Tanner code is a conical combination of normalized characteristic vectors of simple cycles in the Tanner graph (Lemma~\ref{lemma:cycleDecomposition}). Then we show that every cycle of length $\ell$ is a conical combination of $\ell$ deviations in $\calB^{(h)}$ for any arbitrary $h$, which implies Corollary~\ref{cor:path_decomposition}.

\begin{lemma} [simple cycles decomposition]
  \label{lemma:cycleDecomposition} Let $\calC(G)$ denote an even
  Tanner code, and let $\Gamma$ denote the set of simple cycles in
  $G$. For every codeword $x \neq 0^N$, there exists a distribution
  $\rho$ over the set $\Gamma$ and an $\alpha>1$, such that
\begin{equation}
x = \alpha \cdot \E_{\gamma\in_\rho\Gamma}\big[\chi_G(\gamma)\big].
\end{equation}
\end{lemma}
\begin{proof}
Let $\calV_x\triangleq\{v\mid x_v=1\}$, and let $G_x$ denote the subgraph of the Tanner graph $G$ induced by $\calV_x\cup\calN_G(\calV_x)$. Because $x$ is a codeword in an even Tanner code, the degree of every node (both variable nodes and local-code nodes) in $G_x$ is even. Therefore, each connected component in $G_x$ is Eulerian. Denote by $\{G_x^{(j)}\}$ the set of connected components of $G_x$,
and let $\psi^{(j)}$ denote an Eulerian cycle in $G_x^{(j)}$.

Consider a variable node $v$ in the connected component $G_x^{(j)}$,
then the multiplicity of $v$ in $\psi^{(j)}$ equals
$\frac{\deg_G(v)}{2}$. Therefore, $2\cdot\sum_j
[\chi_G(\psi^{(j)})]_v=1$ (and, by definition, $x_v=1$).

Every Eulerian cycle $\psi^{(j)}$ can be decomposed into a set of
edge disjoint simple cycles.  Let $\Gamma^{(j)}$ denote the decomposition of
$\psi^{(j)}$ into simple cycles. Then,
$\chi_G(\psi^{(j)})=\sum_{\gamma\in\Gamma^{(j)}}\chi_G(\gamma)$. Thus,
\begin{equation*}
x_v = 2\cdot\sum_j\sum_{\gamma\in\Gamma^{(j)}}\chi_G(\gamma).
\end{equation*}
Let $\rho$ denote the uniform distribution over $\cup_j\Gamma^{(j)}$
and let $s\triangleq\abs{\cup_j\Gamma^{(j)}}$.  Then,
\[
x = 2 s\cdot\E_{\gamma\in_{\rho} \Gamma} \big[\chi_G(\gamma)\big].
\]
\end{proof}

\begin{corollary}\label{cor:path_decomposition}
  Let $\calC(G)$ denote an even Tanner code, and let $h\in\N_+$. For
  every codeword $x \neq 0^N$, there exists a distribution $\rho$ over
  the set $\calB^{(h)}$ and an $\alpha'>1$ such that \[x
  =\alpha'\cdot\E_{\beta\in_{\rho}{\calB^{(h)}}}[\beta].\]
\end{corollary}

\begin{proof}
  Following Lemma~\ref{lemma:cycleDecomposition}, it suffices to show
  that, for every simple cycle $\gamma$, the set $\{\psi_i\}$
  of paths in $\gamma$ of length $h$
satisfy
\begin{align}
\exists \delta_{\gamma}\geq 1:   \chi_G(\gamma) &=
  \delta_\gamma\cdot\sum_{i=0}^{|\gamma|-1}\frac{1}{h+1}\cdot\chi_G(\psi_i).
\label{eq:chi}
\end{align} Indeed, let $\beta_i(\gamma)\eqdf
\frac{1}{h+1}\cdot\chi_G(\psi_i)$, then
  \begin{align*}
    x &=\alpha \cdot \sum_{\gamma} \rho(\gamma)\cdot \chi_G(\gamma)\\
    &=\alpha \cdot \sum_{\gamma} \left(
      \rho(\gamma)\cdot \delta_\gamma\cdot\sum_i \beta_i(\gamma)\right)\\
    &=\alpha \cdot \sum_{\beta} \left( \rho(\gamma_\beta) \cdot \delta_\gamma \cdot \beta \right).
  \end{align*}
  Note that in the last line, we use the fact that the cycles in $G_x$
  are decomposed into edge disjoint cycles, and hence each deviation
  $\beta$ appears in exactly one cycle, denoted by $\gamma_\beta$. The
  corollary follows because the coefficients $\rho(\gamma_\beta)\cdot
  \delta_\gamma$ are nonnegative and their sum is at least one.

  We now prove Equation~(\ref{eq:chi}). Let
  $\gamma=(v_0,v_1,\ldots,v_{\ell-1},v_{\ell}=v_0)$ be a simple cycle
  in $G$ of length $\ell$. For every $0\leqslant i\leqslant\ell-1$,
  let $\psi_i = (v_i , v_{i+1\mod \ell}, ... , v_{i+h\mod\ell})$
  denote a segment of $\gamma$ that starts at node $v_i$ and contains
  $h$ edges. (Note that if $h\geqslant\girth(G)$, then a single
  segment may traverse a node in the cycle more than once.)  For every
  $0\leqslant j\leqslant h$, a node $v$ appears exactly once as the
  $j$th node in one of the paths $\{\psi_i\}_{i=0}^{\ell-1}$.

  If $h$ is not a multiple of $\ell$, then the multiplicity of every
  vertex $v\in\gamma$ in $\cup_{i=0}^{\ell-1}\psi_i$ equals $h+1$.
  Therefore,
\[\chi_G(\gamma) = \sum_{i=0}^{\ell-1}\frac{1}{h+1}\cdot\chi_G(\psi_i).\]
Otherwise, $\ell$ divides $h$, and every $\psi_i$ is a cycle whose
both endpoints $v_i$ are counted as one occurrence in
$\chi_G(\psi_i)$. Hence, the multiplicity of every vertex $v\in\gamma$
in $\cup_{i=0}^{\ell-1}\psi_i$ equals $h$, and
\[\chi_G(\gamma) = \frac{h+1}{h}\sum_{i=0}^{\ell-1}\frac{1}{h+1}\cdot\chi_G(\psi_i).\]
\end{proof}

\begin{theorem}[local-optimality is sufficient for ML]\label{thm:MLsufficient}
Let $\calC(G)$ denote an even Tanner code. Let $\lambda\in\R^N$ denote the LLR vector received from the channel and let $h\in\N_+$. If $x$ is $h$-locally optimal codeword w.r.t. $\lambda$, then $x$ is also the unique maximum-likelihood codeword w.r.t. $\lambda$.
\end{theorem}

\begin{proof}
The proof follows \cite[proof of Theorem 2]{ADS09} and \cite[proof of Theorem 6]{HE11}.
We use the decomposition proved in Corollary~\ref{cor:path_decomposition} to show that for every codeword $x' \neq x$, $\langle \lambda , x'\rangle > \langle \lambda , x \rangle$.
Let $z \triangleq x \oplus x'$. By linearity, $z\in\calC(G)$. Moreover, $z\neq0^N$ because $x\neq x'$.
By Corollary~\ref{cor:path_decomposition} there exists a distribution over the set $\mathcal{B}^{(h)}$, such that
$\E_{\beta \in \mathcal{B}^{(h)}} \beta = \delta\cdot z$, where $\delta\triangleq\frac{1}{\alpha'}<1$.
Let $f:[0,1]^N \rightarrow \R$ be the affine linear function defined by $f(\beta) \triangleq \langle \lambda , x \oplus \beta \rangle = \langle\lambda , x \rangle + \sum_{i=1}^{N}(-1)^{x_i}\lambda_i \beta_i$.
Then,
\begin{eqnarray*}
\langle \lambda , x \rangle &<& \E_{\beta \in \calB^{(h)}} \langle \lambda , x \oplus \beta \rangle \ \ \ (\text{by local-optimality of $x$}) \\
&=& \langle\lambda,x\oplus\E_{\beta \in \calB^{(h)}}\beta\rangle \ \ \ (\text{by linearity of $f$ and linearity of expectation}) \\
&=& \langle\lambda,x\oplus\delta z \rangle  \ \ \ \ \ \ \ \ \ \ \ \ \ \ (\text{by Lemma \ref{lemma:cycleDecomposition}})\\
&=& \langle \lambda , (1-\delta)x+\delta(x\oplus z) \rangle\\
&=& \langle \lambda , (1-\delta)x+\delta x' \rangle\\
&=& (1-\delta)\langle \lambda , x \rangle+\delta\langle \lambda , x' \rangle.
\end{eqnarray*}
which implies that $\langle \lambda , x' \rangle>\langle \lambda , x \rangle$ as desired.
\end{proof}

\begin{remark}
  Lemma~\ref{lemma:cycleDecomposition} allows one to define
  local-optimality with respect to deviations induced by simple
  cycles.  Local-optimality based on backtrackless paths of arbitrary
  length $h$ decouples the definition of local-optimality from the
  girth of the Tanner graph. The implication of this decoupling on
  iterative decoding is discussed in Section~\ref{sec:discussion}.
\end{remark}

\subsection{Local-Optimality Implies LP-Optimality} \label{subsec:LOimpliesLP}
In the following section we show that local-optimality is sufficient
for LP-optimality (Theorem~\ref{thm:LPsufficient}).  We consider graph
cover decoding introduced by Vontobel and Koetter~\cite{VK05} and its
extension to Tanner codes~\cite[Chapter 2.6]{Hal11}. The proof of
Theorem~\ref{thm:LPsufficient} is based on Lemma~\ref{lemma:cover Opt}
that states that local-optimality is preserved under lifting to any
$M$-cover graph.

The presentation in this section uses the terms and notation of
Vontobel and Koetter~\cite{VK05} (see also~\cite[Chapter 2.6]{Hal11}).  Let $\tilde{G}$ denote an $M$-cover
of $G$.  Let $\tilde{x}=x^{\uparrow M}\in \mathcal{C}(\tilde{G})$ and
$\tilde{\lambda}=\lambda^{\uparrow M} \in \R^{N\cdot M}$ denote the
$M$-lifts of $x$ and $\lambda$, respectively.

\begin{proposition}[local-optimality of the all-zero codeword is preserved by $M$-lifts] \label{prop:coverOptimality}
$0^N$ is $h$-locally optimal codeword w.r.t.
  $\lambda \in \R^N$ if and only if $0^{N\cdot M}$ is $h$-locally optimal codeword w.r.t. $\tilde{\lambda}$.
\end{proposition}
\begin{proof}
Consider the surjection $\varphi$ of paths with length $h$ in $\tilde{G}$ to paths in $G$. This surjection is based on the covering map between $\tilde G$ and $G$.
Given a path $\tilde{\psi}$ in $\tilde{G}$, let $\psi\triangleq\varphi(\tilde{\psi})$. Let $\tilde\beta\triangleq\frac{1}{h+1}\chi_{\tilde{G}}(\tilde\psi)$ and $\beta\triangleq\frac{1}{h+1}\chi_G(\psi)$. The proposition follows because $\langle \lambda, \beta
  \rangle =\langle \tilde\lambda, \tilde\beta \rangle$.
\end{proof}

\noindent
The following lemma states that local-optimality is preserved by lifting to an $M$-cover.
\begin{lemma}\label{lemma:cover Opt}
$x$ is $h$-locally optimal w.r.t. $\lambda$ if and only if
$\tilde x$ is $h$-locally optimal w.r.t. $\tilde\lambda$.
\end{lemma}

\begin{proof}
Assume that $\tilde{x}$ is $h$-locally optimal codeword w.r.t. $\tilde{\lambda}$.  By Corollary~\ref{corr:LOsymmetry}, $0^{N\cdot M}$ is $h$-locally optimal w.r.t. $(-1)^{\tilde{x}} \ast \tilde{\lambda}$.  By Proposition~\ref{prop:coverOptimality}, $0^N$ is $h$-locally optimal w.r.t. $(-1)^x\ast \lambda$. By Corollary~\ref{corr:LOsymmetry}, $x$ is $h$-locally optimal w.r.t. $\lambda$. Each of these implications is necessary and sufficient, and the lemma follows.
\end{proof}

The following theorem is obtained as a corollary of
Theorem~\ref{thm:MLsufficient} and Lemma~\ref{lemma:cover Opt}.
The proof is based on arguments utilizing properties of graph cover decoding. Those arguments are used for a reduction from ML-optimality to LP-optimality~\cite[Theorem 8]{HE11}.

\begin{theorem}[local-optimality is sufficient for LP optimality]\label{thm:LPsufficient}
If $x$ is a $h$-locally optimal codeword w.r.t. $\lambda$, then $x$ is also the unique optimal LP solution given $\lambda$.
\end{theorem}

\begin{proof}
  Suppose that $x$ is $h$-locally optimal codeword w.r.t. $\lambda \in
  \R^N$. By~\cite[Proposition 10]{VK05}, for every basic
  feasible solution $z \in [0,1]^N$ of the LP, there exists an
  $M$-cover $\tilde{G}$ of $G$ and an assignment $\tilde{z}\in
  \{0,1\}^{N\cdot M}$ such that $\tilde{z} \in \mathcal{C}(\tilde{G})$
  and $z = \zeta(\tilde{z})$, where $\zeta(\tilde{z})$ is the image of
  the scaled projection of $\tilde{z}$ in $G$ (i.e., the
  pseudo-codeword associated with $\tilde{z}$). Moreover, since the
  number of basic feasible solutions is finite, we conclude that there
  exists a finite $M$-cover $\tilde{G}$ such that every basic feasible
  solution of the LP admits a valid assignment in $\tilde{G}$.

Let $z^*$ denote an optimal LP solution given $\lambda$. Without loss of generality $z^*$ is a basic feasible solution. Let $\tilde{z}^* \in \{0,1\}^{N\cdot M}$ denote the $0-1$ assignment in the $M$-cover $\tilde{G}$ that corresponds to $z^* \in [0,1]^N$. By~\cite[Proposition 10]{VK05} and the optimality of $z^*$ it follows that $\tilde{z}^*$ is a codeword in $\mathcal{C}(\tilde{G})$ that minimizes $\langle \tilde{\lambda} , \tilde{z} \rangle$ for $\tilde{z} \in \mathcal{C}(\tilde{G})$, namely $\tilde{z}^*$ is the ML codeword in $\calC(\tilde{G})$ w.r.t. $\lambda^{\uparrow M}$.

Let $\tilde{x}=x^{\uparrow M}$ denote the $M$-lift of an $h$-locally optimal codeword $x$. Note that because $x$ is a codeword, i.e., $x \in \{0,1\}^N$, there is a unique pre-image of $x$ in $\tilde{G}$, which is the $M$-lift of $x$. Lemma~\ref{lemma:cover Opt} implies that $\tilde{x}$ is $h$-locally optimal codeword w.r.t. $\tilde{\lambda}$. By Theorem~\ref{thm:MLsufficient}, we
also get that $\tilde{x}$ is the ML codeword in $\calC(\tilde{G})$ w.r.t. $\lambda^{\uparrow M}$. Moreover, Theorem~\ref{thm:MLsufficient} guarantees the uniqueness of an ML
optimal solution. Thus, $\tilde{x} = \tilde{z}^*$.
Because $\tilde{x} = \tilde{z}^*$, when projected to
$G$, we get that $x = z^*$ and
uniqueness follows, as required.
\end{proof}

\section{Probabilistic Analysis of Path-Based Local-Optimality} \label{sec:bounds}

In the previous section, we showed that LP-decoding succeeds if a locally-optimal codeword exists w.r.t. the received LLR. In this section we analyze the probability that a locally-optimal codeword exists for even Tanner codes in MBIOS channels.
The following equation justifies the all-zero codeword assumption for analyses based on local-optimality characterizations.
\begin{align}\label{eqn:LPfailure}
\Pr\{\mathrm{LP\ decoding\ fails}\} &\triangleq\Pr\{x\neq\hat{x}^{\LP}(\lambda)\mid c=x\} \nonumber\\
&\stackrel{(1)}{\leqslant}\Pr\big\{x\mathrm{\ is\ not\ } h\mathrm{-locally\ optimal\ w.r.t.\ }\lambda \big| c=x\big\}\nonumber\\
&\stackrel{(2)}{=}\Pr\big\{0^N\mathrm{\ is\ not\ } h\mathrm{-locally\ optimal\ w.r.t.\ }(-1)^x\ast\lambda \big| c=x\big\}\nonumber\\
&\stackrel{(3)}{=}\Pr\big\{0^N\mathrm{\ is\ not\ } h\mathrm{-locally\ optimal\ w.r.t.\ }\lambda \big| c=0^N\big\}\nonumber\\
&\stackrel{(4)}{=}\Pr\big\{\exists \beta \in \calB^{(h)} \mathrm{\ such\ that\ } \langle \lambda, \beta \rangle \leqslant 0 \big| c=0^N\big\}.
\end{align}
Inequality~(1) is the contrapositive statement of
Theorem~\ref{thm:LPsufficient}. Equality~(2) follows
Corollary~\ref{corr:LOsymmetry}. For MBIOS channels,
$\Pr(\lambda_i\mid c_i=0)=\Pr(-\lambda_i\mid c_i=1)$. Therefore, the
mapping $(x,\lambda)\mapsto(0^N,b\ast\lambda)$ where
$b_i\triangleq(-1)^{x_i}$ is a measure preserving mapping.
Equality~(3) follows by applying this mapping to
$(x,b\ast\lambda)\mapsto(0^N,b\ast b\ast\lambda)$. Equality~(4) follows by the
definition of path-based local-optimality.

Following Equation~(\ref{eqn:LPfailure}), our goal is to prove an upper bound on the probability that there exists a path of length $h$ in $G$ whose normalized characteristic vector has non-positive cost w.r.t. $\lambda$.

We use the following notation.
For a path $\psi$, the \emph{normalized cost of $\psi$ w.r.t. $\lambda$} is defined by $\cost(\psi)\triangleq\langle\lambda,\chi_G(\psi)\rangle$.
Let $d_L^{\min}\triangleq\min\{deg_G(v)\mid v\in\calV\}$, $d_L^{\max}\triangleq\max\{deg_G(v)\mid v\in\calV\}$, and $d_R^{\max}\triangleq\max\{deg_G(C)\mid C\in\calJ\}$. Let $D\triangleq d_L^{\max}\cdot d_R^{\max}$.

The intuition of the analysis is that long simple paths are unlikely to have non-positive cost.
We restrict the path length by $h<\girth(G)$ only for the probabilistic analysis. Tanner graphs with logarithmic girth can be constructed explicitly (see e.g.~\cite{Gal63}).
In particular, we assume that $\girth(G)>\log_D(N)$.

The following theorem presents an analytical
bound on the word error probability of the LP decoder over the BSC.
\begin{theorem} \label{thm:BSCbound}
Let $\calC(G)$ denote an even Tanner code of length $N$ such that $g\triangleq\log_D(N)<\girth(G)$.
Consider a BSC with crossover probability $p$. For any $\epsilon > 0$, if $p <
D^{-2\cdot(1+\frac{d_L^{\min}}{d_L^{\max}})\cdot(\epsilon + \frac{3}{2}+\frac{1}{2}\log_D(2))}$, then the LP decoder fails to decode the transmitted codeword with a probability of at most $N^{-\epsilon}$.
\end{theorem}
For example, for left-regular codes (i.e., $d_L^{\max}=d_L^{\min}$), set $\epsilon=\frac{1}{2}$. If $p\leqslant\frac{D^{-8}}{4}$, then the word error probability of LP-decoding is at most $\frac{1}{\sqrt{N}}$.
\begin{proof}
By Equation~(\ref{eqn:LPfailure}) we may assume that the all-zero codeword is transmitted, i.e., $c=0^N$. Hence $\lambda_v=1$ w.p. $(1-p)$ and $\lambda_v=-1$ w.p. $p$.
We bound the word error probability $P_w$ using a union bound over all the events where simple paths of length $g$ have non-positive normalized cost in $G$ w.r.t. $\lambda$.

Let $\psi$ be a particular path of length $g$; $\psi$ contains $\frac{g}{2}$ variable nodes. Each variable node in the path is assigned $+1$ with probability $(1-p)$ and $-1$ with probability $p$. Because of the degree normalization, at least $\frac{g}{2}\cdot\frac{d_L^{\min}}{d_L^{\min}+d_L^{\max}}$ of the variable nodes in $\psi$ must be assigned $-1$ to obtain
$\cost(\psi) \leqslant 0$. Let $\delta\triangleq\frac{d_L^{\min}}{d_L^{\min}+d_L^{\max}}$. Therefore,
\begin{equation} \Pr\{\cost(\psi)\leqslant0\}\leqslant\binom{\frac{g}{2}}{\frac{g}{2}\cdot\delta} p^{\frac{g}{2}\cdot\delta}\leqslant2^{\frac{g}{2}}\cdot p^{\frac{g}{2}\cdot\delta}.
\end{equation}
There are at most $|\calV|
\cdot D^{\frac{g}{2}}$ different simple paths of length $g$ in $G$. By the union bound,
\begin{equation}
\begin{aligned}
P_w & \leqslant \abs{\calV}\cdot D^{\frac{g}{2}}\cdot2^\frac{g}{2}\cdot p^{\frac{g}{2}\cdot\delta}\\
& \leqslant N\cdot D^{\frac{1}{2}\log_D(N)}\cdot2^{\frac{1}{2}\log_D(N)}\cdot p^{\frac{1}{2}\log_D(N)\cdot\delta}\\
&\leqslant N\cdot N^{\frac{1}{2}}\cdot N^{\frac{1}{2}\log_D(2)}\cdot N^{\frac{1}{2}\log_D(p)\cdot\delta}\\
&=N^{\frac{3}{2}+\frac{1}{2}\log_D(2)+\frac{1}{2}\log_D(p)\cdot\delta}\leqslant N^{-\epsilon}.
\end{aligned}
\end{equation}
\end{proof}

For the case of BI-AWGN channel, we derive a bound on the word error probability for left-regular even Tanner codes, i.e., $d_L=d_L^{\min}=d_L^{\max}$.
The extension to the case of irregular even Tanner codes requires exhaustive notation and computations.

\begin{theorem}\label{thm:AWGNbound}
Let $\calC(G)$ denote a left-regular even Tanner code of length $N$ such that $g=\log_D(N)$ where $D\triangleq d_L\cdot d_R^{\max}$.
Consider a BI-AWGN channel with variance $\sigma^2$. For any $\epsilon > 0$, if $\sigma^2 <
\frac{\log_D (e)}{6 + 4\epsilon}$, then the LP decoder fails to decode the transmitted codeword with a probability of at most $\frac{\sigma}{\sqrt{\pi\log_D(N)}}\cdot N^{-\epsilon}$.
\end{theorem}

\begin{proof}
By Equation~(\ref{eqn:LPfailure}) we assume that the all-zero codeword is transmitted, i.e., $c=0^N$. Hence $\lambda_v=1 + \phi_{i}$ where $\phi_{i} \sim \mathcal{N}(0,\sigma^{2})$ is a zero-mean Gaussian random variable with variance $\sigma^{2}$.
We bound the
word error probability $P_w$ using a union bound over all the events where simple paths of length $g$ have non-positive cost in $G$.

Let $\psi$ be a particular path of length $g$.
If the sum of the costs of the variable nodes in $\psi$ is non-positive, then $\cost(\psi)\leqslant0$.
Hence,
\begin{align}
  \begin{split}
    Pr \big\{ \cost(\psi) \leqslant 0 \big\} &= Pr \big\{ \sum_{i=1}^{\frac{g}{2}}{(1 + \phi_{i})} \leqslant 0 \big\} \\
    &= Pr \big\{ \sum_{i=1}^{\frac{g}{2}} {\phi_{i} \leqslant -\frac{g}{2}} \big\}.
  \end{split}
\end{align}
The sum of independent Gaussian random variables (RVs) with zero mean is a zero-mean Gaussian RV whose variance equals to the sum of the variances of the accumulated variables.  Let $\Phi =
\sum_{i=1}^{\frac{g}{2}}{\phi_{i}}$, then $\Phi \sim
\mathcal{N}(0,\sigma^{2}\cdot\frac{g}{2})$ is a zero-mean Gaussian RV with variance
$\sigma^{2}\cdot\frac{g}{2}$. Moreover, the Gaussian distribution function is symmetric around $0$.  Therefore,
\begin{align}
  \begin{split}
    Pr \{ \cost(\psi) \leqslant 0 \} &= Pr \big\{ \Phi \leqslant -\frac{g}{2} \big\} \\
    &= Pr \big\{ \Phi \geqslant \frac{g}{2} \big\}.
  \end{split}
\end{align}

For a Gaussian RV $\phi \sim \mathcal{N}(0,\sigma^{2})$ with zero
mean and variance $\sigma^{2}$, the inequality
\begin{equation} \label{eq:double}
Pr\{ \phi \geqslant x \} \leqslant \frac{\sigma}{x\sqrt{2\pi}}
e^{-\frac{x^{2}}{2\sigma^{2}}}
\end{equation}
holds for every $x>0$ \cite{Fel68}. We conclude that
\begin{equation}
  Pr \{ \cost(\psi) \leqslant 0 \} \leqslant
    \frac{\sigma}{\sqrt{\pi g}} e^{-\frac{g}{4\sigma^{2}}}.
\end{equation}

There are at most $\abs{\calV}
\cdot D^{\frac{g}{2}}$ different simple paths of length $g$ in $G$. By the union bound,
\begin{equation}
\begin{aligned}
P_w &\leqslant \abs{\calV}\cdot D^{\frac{g}{2}}\cdot\Pr \{ \cost(\psi) \leqslant 0\} \\
&\leqslant N\cdot D^{\frac{g}{2}} \frac{\sigma}{\sqrt{\pi g}} e^{-\frac{g}{4\sigma^{2}}}\\
&\leqslant N^{\frac{3}{2}}\cdot \frac{\sigma}{\sqrt{\pi \log_D(N)}} e^{-\frac{\log_D(N)}{4\sigma^{2}}}\\
&= \frac{\sigma}{\sqrt{\pi \log_D(N)}}N^{\frac{3}{2}-\frac{1}{4\sigma^{2}}\log_D(e)}\\
&\leqslant \frac{\sigma}{\sqrt{\pi \log_D(N)}}\cdot N^{-\epsilon}.
\end{aligned}
\end{equation}

\end{proof}

\section{Discussion}\label{sec:discussion}

\subsection{Finding Locally-Optimal Codewords with Message-Passing Algorithm}

A message-passing decoding algorithm, called normalized weighted
min-sum (\NWMS), was presented in~\cite{EH11} for Tanner codes with
single parity-check (SPC) local codes. The \NWMS\ decoder is
guaranteed to compute the ML-codeword in $h$ iterations provided that
a locally-optimal codeword with height $h$ exists. The
local-optimality characterization in~\cite{EH11} is stronger and is based
on subtrees of computation trees of the Tanner graph. We note that for
even Tanner codes with SPC local codes, the \NWMS\ decoder with
uniform weights (i.e., $w=1^h$) is guaranteed to compute the
ML-codeword in $h$ iterations provided that a path-based
$h$-locally optimal codeword exists. The number of iterations $h$ may
exceed the girth of the Tanner graph.

Consider an even Tanner code and an LLR vector $\lambda$.
Assume that there exists a codeword $x$ that is $h$-locally optimal
w.r.t. $\lambda$.  Given this assumption, maximum-likelihood decoding
is equivalent to solving the following problem: Find an
assignment $x\in\{0,1\}^N$ to variable nodes such that every path of
length $h$ in $G$ w.r.t. vertex weights
$(-1)^{x_v}\cdot\lambda_v/\deg_G(v)$ has positive weight.  By
Theorem~\ref{thm:LPsufficient}, LP-decoding computes such a valid
assignment $x$. The iterative \NWMS\ decoding algorithm also computes
such a valid assignment $x$ provided that the local codes are
restricted to SPC codes.

\subsection{Punctured Tanner Codes}

Local-optimality characterization and its analysis remain valid under
puncturing of a codeword. Puncturing of
a code is specified by a subset of variable nodes that are not transmitted; this subset is called the \emph{puncturing
  pattern}. A punctured code can be analyzed simply by
zeroing the LLR values of the punctured variable nodes.

How does puncturing affect the probability that a locally-optimal
codeword exists?  Consider a probabilistic analysis for
local-optimality that is based on a simple union bound over the set of
deviations of a Tanner code. We say that a variable node $v$
\emph{participates} in deviation $\beta$ if $\beta_v\neq0$. Suppose
that the puncturing pattern is chosen so that for every deviation
$\beta$, at most a constant fraction $\alpha$ of participating
variable nodes are punctured. Then, the same union bound analysis can
be applied directly to the punctured Tanner code. Such an analysis
implies bounds that are similar to the bounds obtained for the
(unpunctured) Tanner code; the only difference is that some parameters
need to be scaled by a constant factor that is proportional to $\alpha$.

\subsection{Analysis of Repeat-Accumulate Codes via Local-Optimality}

Feldman and Karger~\cite{FK02,FK04} introduced the concept of
linear-programming based decoding for repeat-accumulate RA($q$)
turbo-like codes. In an RA($q$) codes: (1)~an information word is
repeated $q$ times, (2)~the repeated information word is permuted by
an interleaver, and (3)~codeword bit $i$ equals to the parity of the
sum of the first $i$ bits (prefix) of the permuted repeated
information word. For RA($2$) codes over the BSC, they proved that the
word error probability of LP-decoding is bounded by an inverse
polynomial in the code length, given a certain constant threshold on
the noise. A similar claim was also proved for the BI-AWGN channel.
These bounds were further improved based on exact combinatorial
characterization of an error event and a refined algorithmic
analysis~\cite{HE05}.

Recently, Goldenberg and Burshtein~\cite{GB11} generalized the
analysis of Feldman and Karger~\cite{FK04} to RA($q$) codes with even
repetition $q\geqslant4$. For this family of codes, they proved
inverse polynomial bounds in the code length on the word error
probability of LP-decoding. These bounds are based on analyzing complicated graph structures, called
hyper-promenades in hypergraphs, that were defined in~\cite{FK04}.

In fact, the same bounds can be obtained by local-optimality with
deviations induced by short paths. The idea is to consider
(non-systematic) irregular repeat-accumulate codes with even
repetition factors as punctured Tanner codes (as illustrated in
Figure~\ref{fig:RA-TannerGraph}). Namely, set the LLR of each
systematic variable node to zero. Notice that, for every backtrackless
path in Tanner graphs of repeat accumulate codes, at most half of the
variable nodes are systematic. Therefore, at most half of the variable
nodes in each deviation of path-based local-optimality are punctured.
We can now analyze the word error probability by replacing
$\frac{g}{2}$ by $\frac{g}{4}$ in the proofs of Theorems~\ref{thm:BSCbound}
and~\ref{thm:AWGNbound}. Hence the unified technique of
local-optimality with deviations induced by short paths yields the
same results proved by Feldman and Karger~\cite{FK04} and Goldenberg
and Burshtein~\cite{GB11}.
Improving these bounds on successful decoding of RA($q$) codes (for
$q\geq 3$) remains an intriguing open question.

\begin{figure}
  \begin{center}
 \includegraphics[width=.9\textwidth]{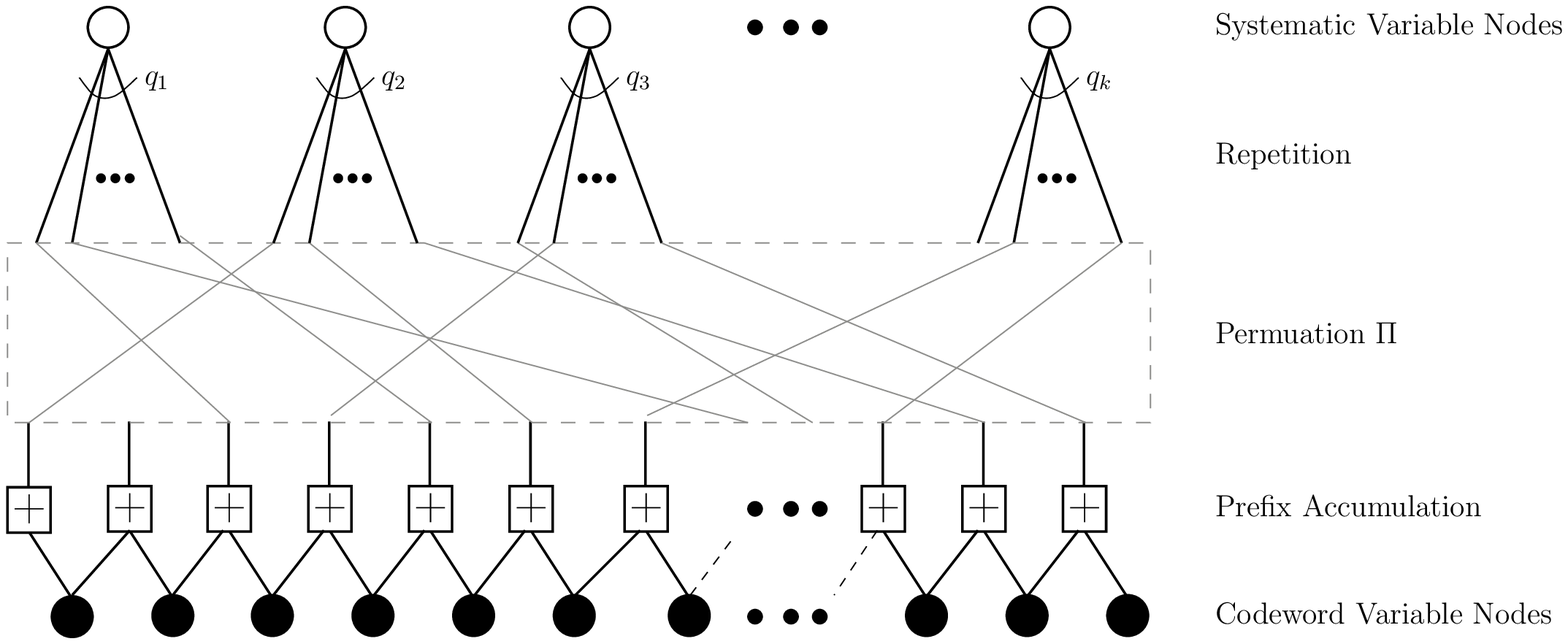}
 \caption{Tanner graph representation of an irregular repeat-accumulate code. \newline{\small The variable nodes in the top row of nodes (illustrated by white circles) correspond to the $k$ systematic bits. The repetition of every systematic bit corresponds to the degree $q_i$ of each systematic variable node (we assume that $q_i$ is even for every $i$). The interleaving process corresponds to the permutation of edges. The prefix accumulation corresponds to the alternating chain of single-parity check code nodes (illustrated by plus-squares) and codeword variable nodes (denoted by black circles) in the bottom.}}
  \label{fig:RA-TannerGraph}
  \end{center}
\end{figure}

\section{Conclusions}\label{sec:conclusions}
We present a simple application of the proof technique in~\cite{ADS09}
for bounds on word error probability with LP-decoding.  The set of
deviations used for defining local-optimality is induced by paths in
the Tanner graph. We apply the proof technique to a family of codes,
called even Tanner codes, that contains repeat accumulate codes with a
even repetition factors, LDPC codes with even left degrees, and
expander codes with even variable node degrees and even weighted local
codes. Inverse polynomial error bounds are proved for these codes for
the BSC and AWGN channel.

Stronger error bounds have been obtained for LDPC
codes~\cite{KV06,ADS09,HE11} and Tanner codes~\cite{EH11}
(without the restriction to even degrees and even weighted local
codes) by considering more complicated graphical structures and a more
sophisticated analysis. In these cases, inverse exponential error
bounds and improved bounds on noise thresholds were presented for
regular codes whose Tanner graphs have logarithmic girth. For example,
the local-optimality characterization for Tanner codes presented
in~\cite{EH11} is based on projections of weighted subtrees in
computation trees of the Tanner graph. The error bounds in this case
are based on an analysis of a sum-min-sum random process on trees.
While this local-optimality characterization applies to any regular
and irregular Tanner code, the probabilistic analysis and the
error bounds were restricted to regular Tanner codes.  The
simplicity of local-optimality based on paths enables us to obtain
(weak) bounds even for irregular codes.  Two interesting open
questions related to proving stronger bounds on the error probability
are (i)~extend the analysis of inverse exponential bounds to irregular
Tanner graphs, and (ii)~obtain bounds with respect to local-optimality
even ``beyond the girth''.

\bibliographystyle{alpha}

\newcommand{\etalchar}[1]{$^{#1}$}

\end{document}